\documentclass[letterpaper, 10pt, conference]{ieeeconf}
\IEEEoverridecommandlockouts
\let\labelindent\relax
\usepackage{enumitem}
\usepackage[caption=false,font=footnotesize]{subfig}
\usepackage[utf8]{inputenc}
\usepackage[usenames,dvipsnames,table]{xcolor}
\usepackage{amsmath,amssymb}
\usepackage{hyperref}
\usepackage[noabbrev,capitalize]{cleveref}
\usepackage{accents}
\usepackage{graphicx}
\usepackage{multirow}
\usepackage{pgfplots}
\usepackage{bcheader}
\usepackage{cite}
\usepackage{tikz}
\usetikzlibrary{calc}

\definecolor{mygreen}{RGB}{163,210,163}
\definecolor{myblue}{RGB}{128,128,204}

\definecolor{myred}{RGB}{229,118,106}

\tikzset{%
arrow/.style={very thick},
 stable fill/.style={color=mygreen,fill opacity=1},
 stable draw/.style={thick,color=mygreen!50!black,
 opacity=0
 },
 nash fill/.style={color=myblue,fill opacity=.5},
 nash draw/.style={thick,color=myblue!50!black},
 imaginary fill/.style={color=blue!5,fill opacity=0.5},
 imaginary draw/.style={thick,blue!30},
 real fill/.style={color=myred!20,fill opacity=0.5},
 real draw/.style={myred!50!black,thick}
}

\newif\ifcommentary
\commentarytrue

\IEEEoverridecommandlockouts                              %
\title{\LARGE \bf Stability of Gradient Learning Dynamics in 
Continuous Games:  \\ Scalar Action Spaces}%
\author{%
Benjamin J. Chasnov$^\ast$,  Daniel Calderone$^\ast$, Beh\c cet A\c c\i kme\c se, Samuel A. Burden, Lillian J. Ratliff%
\thanks{B. Chasnov, S. Burden, and L. Ratliff are with the Department of  Electrical and Computer Engineering, 
        University of Washington, Seattle, WA 98115
        {\tt\small $\{$bchasnov,sburden,ratliffl$\}$@uw.edu}}
        \thanks{D. Calderone and B. A\c c\i kme\c se are with the Department of Aeronautics and Astronautics,  University of Washington, Seattle, WA 98115
        {\tt\small $\{$djcal,behcet$\}$@uw.edu}}
\thanks{$^\ast$ Denotes equal contribution.}\thanks{Funding for this work is provided by
NSF Award \#1836819
and NIH 5T90DA032436-09.}
}

\newif\ifsixpages
\sixpagesfalse %

\pgfplotsset{compat=1.16}
\begin{document}

\maketitle
\thispagestyle{empty}
\pagestyle{empty}
\begin{abstract}
 \noindent
 Learning processes in games
explain how players grapple with one another in seeking an equilibrium.
We study a natural model of 
learning 
based on individual gradients in
two-player continuous games. 
In such games, the arguably natural notion of a local equilibrium is a differential Nash equilibrium. 
However, 
the set of locally exponentially stable equilibria 
of the learning dynamics
do not necessarily coincide with 
the set of differential Nash equilibria of the corresponding game.
To characterize this gap, we provide formal guarantees for the stability or instability of such fixed points by leveraging the spectrum of the linearized game dynamics. 
We provide 
a comprehensive understanding of scalar games and 
find that 
equilibria that are both stable and Nash 
are robust to variations in learning rates. 

\end{abstract}

\section{Introduction}
 The study of learning in games 
 is experiencing a resurgence in
 the control theory \cite{ratliff2016characterization,tang2019distributed,tatarenko2018learning},
 optimization \cite{mazumdar2018fundamental,mertikopoulos2019learning}, and 
 machine learning \cite{bu2019global,chasnov2019convergence,goodfellow2014gans,metz2016unrolled,fiez2019convergence} communities. 
 Partly driving this resurgence 
 is the prospect
 for game-theoretic
 analysis to yield
 machine learning algorithms
 that generalize better or are more robust.
 Towards understanding the optimization landscape in such formulations, dynamical systems theory is emerging as a principal tool for analysis and ultimately synthesis~\cite{mazumdar2018fundamental,boone2019darwin,mertikopoulos2018cycles,berard2019closer,balduzzi2020smooth}.
A predominant learning paradigm used across these different domains is gradient-based learning. Updates in large decision spaces can be performed locally with minimal information, while still guaranteeing local convergence in many problems~\cite{chasnov2019convergence,mertikopoulos2019learning}.

One of the 
primary means to understand the 
optimization landscape 
of games is the eigenstructure and spectrum of the Jacobian of the learning dynamics in a neighborhood of a 
stationary point.
In particular, for a zero-sum continuous game $(\cost,-\cost)$ with some continuously-differentiable 
$\cost$,
the Nash equilibria are saddle points of the function $\cost$. As the example in Fig.~\ref{fig:zssaddle} demonstrates, not all saddle points are relevant. Loosely speaking, the equilibrium conditions for the game correspond to constraints on the curvature directions of the cost function and hence, on the eigenstructure of the Jacobian nearby equilibria.

The local stability of a hyperbolic 
 fixed point in a non-linear system can be assessed by examining the
 eigenstructure of the linearized dynamics~\cite{sastry1999nonlinear,khalil2002nonlinear}. However, in a game context there are extra constraints coming from the underlying game---that is, players are constrained to move only along directions over which they have control. 
 They can only control their individual actions, as opposed to the entire state of the dynamical system corresponding to the learning rules being applied by the agents. 
 It has been observed in earlier work that not all stable attractors of gradient play are local Nash equilibria and not all local Nash equilibria are stable attractors of gradient play~\cite{mazumdar2018fundamental}. Furthermore, changes in players' learning rates---which corresponds to scaling rows of the Jacobian---can change an equilibrium from being stable to unstable and vice versa~\cite{chasnov2019convergence}.

To summarize, there is a subtle but extremely important difference between game dynamics and traditional nonlinear dynamical systems:  alignment conditions are important for distinguishing between equilibria that have game-theoretic meaning versus those which are simply stable attractors of learning rules, and features of learning dynamics such as learning rates can play an important role in shaping not only equilibria  but also alignment properties. Motivated by this observation along with the recent resurgence of applications of learning in games in control, optimization, and machine learning, in this paper we provide an in-depth analysis of the spectral properties of gradient-based learning in two-player continuous games.

 \begin{figure}[t!]
  \input{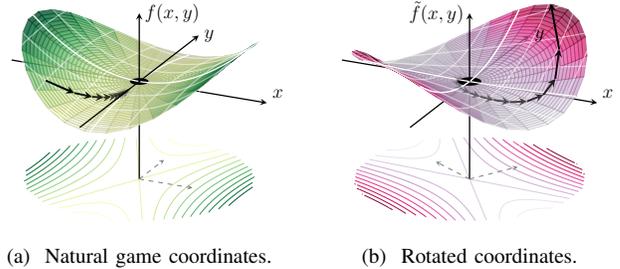}
     \caption{\emph{%
     Cost landscape 
     is crucial to understanding dynamics.} The zero-sum game 
     defined by
     $f(x,y)=\tfrac{1}{2}x^2 -\tfrac{1}{8}y^2$ has a Nash equilibrium at the origin, which is a stable saddle point of gradient play~%
    \eqref{eq:gradientplay}.
     If the cost function is rotated to 
     $\tilde{f}(x,y)=\tfrac{1}{32}x^2 +\tfrac{11}{32}y^2-\tfrac{5\sqrt{3}}{16}xy$---a rotation by $\tfrac{\pi}{3}$---then 
     the origin is no longer a Nash equilibrium, and is \emph{unstable} under gradient play.
     }
      \label{fig:zssaddle}
 \end{figure}
 
\textbf{Contributions.}  
This paper 
characterizes 
the spectral properties of
structured $2\times 2$ matrices
and analyzes the stability of equilibria in continuous games.
Having a complete algebraic understanding of 
the spectrum of the game Jacobian
is fundamental to 
understanding when Nash equilibria coincide with stable equilibria.
Many of our results are geometric in nature and are accompanied by diagrams.

It is known that the quadratic numerical range of a block operator matrix contains the operator's (point) spectrum~\cite{tretter2008spectral}.
Thus, it serves as an important tool for quantifying the spectrum of 
two-player game dynamics.
The method for obtaining the quadratic numerical range
is by reducing a block matrix to $2\times 2$ matrices.

Towards this end, we decompose the $2\times 2$ game Jacobian into coordinates that reflect the interaction between the players. The decomposition provides insights on games and vector fields in general, which permits us to provide a complete characterization of the stability of equilibria in two-player gradient learning dynamics.

\textbf{Organization.}
In Section~\ref{sec:setup},
we describe the gradient-based learning paradigm 
and
analyze the spectral properties of block operator matrices using 
the quadratic numerical range~\cite{tretter2008spectral}.
In Section~\ref{sec:main},
we analyze the spectral properties of two-player continuous games on scalar action spaces.
Our main results are on general-sum games, with insights drawn from specific classes of games.
In Section~\ref{sec:certificates}, we  certify
the stability of Nash and non-Nash equilibria in two-player scalar games. A key finding is that in the scalar case, equilibria that are both stable and Nash are robust to variations in learning rates; in the vector case, they are not.
We provide an example in Section~\ref{sec:examples} and conclude in Section~\ref{sec:conclusion}.

\section{Preliminaries}
\label{sec:setup}

This section contains game-theoretic preliminaries, mathematical formalism, and
a description of the gradient-based learning paradigm studied in this paper.

\subsection{Game-Theoretic Preliminaries}
A $2$-player \emph{continuous game} 
$\mc G = (\cost_1, \cost_2)$ 
is a collection of  costs 
defined on $X=X_1\times X_2$
where player (agent)
$i \in \mc I=\{1,2\}$ 
has cost $f_i:X\to \mb{R}$. In this paper, the results apply to games with sufficiently smooth costs $\cost_i\in C^r(
X,
\mb{R})$ for some $r\geq 0$.
Agent $i$'s set of feasible actions 
is the $\dimm_i$-dimensional precompact set
$X_i \subseteq \mb R^{\dimm_i}$. The notation $x_{-i}$ denotes the action of player $i$'s competitor; 
that is, $x_{-i}=x_j$ where $j\in \mc{I}\backslash\{i\}$.\footnote{For 2-player games, $x_{-1} = x_2$ and $x_{-2} = x_1$.}

The most common and arguably natural notion of an equilibrium in continuous games is due to  Nash~\cite{nash1951non}.
\begin{definition}[Local Nash equilibrium] A joint action profile $x=(x_1,x_2)\in W_1\times W_2\subset X_1\times X_2$ is a local Nash equilibrium on $W_1\times W_2$ if, for each player $i\in\mc{I}$,
$f_i(x_i,x_{-i})\leq f_i(x_i^\prime,x_{-i})$, $\forall x_i^\prime \in W_i$.
\end{definition}
A local Nash equilibrium can equivalently be defined as in terms of best response maps: $x_i\in \arg\min_{y}f_i(y,x_{-i})$.  
From this perspective, local optimality conditions for players' optimization problems give rise to the notion of a differential Nash equilibrium~\cite{ratliff2013characterization,ratliff2016characterization}; non-degenerate differential Nash are known to be generic and structurally stable amongst local Nash equilibria in sufficiently smooth games~\cite{ratliff2014allerton}. Let $D_if_i$ denote the derivative of $f_i$ with respect to $x_i$ and, analogously, let $D_i(D_if_i)\equiv D_i^2f_i$ be player $i$'s individiaul Hessian.
\begin{definition}
\label{def:nash}
For continuous game $\mc{G}=(\cost_1, \cost_2)$ where $\cost_i \in C^2(X_1 \times X_2,\R)$,
 a joint action profile $(x_1,x_2) \in X_1 \times X_2$ is a 
 \emph{differential Nash equilibrium}
 if 
 $D_if_i(x_1, x_2) = 0$ and $D_i^2f_i(x_1,x_2)> 0$ for each $i\in \mc{I}$.
\end{definition}
A differential Nash equilibrium is a strict local Nash equilibrium \cite[Thm.~1]{ratliff2013characterization}. Furthermore, the conditions $D_if_i(x)=0$ and $D_i^2f_i(x)\geq 0$ are necessary for a local Nash equilibrium \cite[Prop.~2]{ratliff2013characterization}.

Learning processes in games, and their study, arose as one of the explanations for how players grapple with one another in seeking an equilibrium~\cite{fudenberg1998theory}. In the case of sufficiently smooth games, gradient-based learning is a natural learning rule for myopic players%
\footnote{A mypoic player effectively believes it cannot influence its opponent's future behavior, and reacts only to local information about its cost.}.

\subsection{Gradient-based Learning as a Dynamical System %
}
At time $\kk$, 
a myopic agent $i$
updates its current action $x_i(\kk)$ by following the gradient of its
individual cost $\obj_i$
given the decisions of its competitors $x_{-i}$.
The synchronous adaptive process that arises is the discrete-time 
dynamical system
\begin{equation}
    \x_i(\kk+1) =x_{i}(\kk)-\gamma_i D_if_i(x_i(\kk),x_{-i}(\kk))
    \label{eq:gradientplay}
\end{equation} for each $i\in \mc{I}$ where $D_if_i$ is the gradient of player $i$'s cost with respect to $x_i$ and $\gamma_i$ is player $i$'s learning rate. 

\paragraph{Stability}
Recall that a matrix $A$ is called Hurwitz if its spectrum lies in the open left-half complex plane $\mb{C}_-^\circ$. Furthermore, we often say such a matrix is \emph{stable} in particular when $A$ corresponds to the dynamics of a linear system $\dot{x}=Ax$ or the linearization of a nonlinear system around a fixed point of the dynamics.\footnote{The Hartman-Grobman theorem~\cite{sastry1999nonlinear} states that around any hyperbolic fixed point of a nonlinear system, there is a neighborhood on which the nonlinear system is stable if the spectrum of Jacobian lies in $\mb{C}_-^\circ$.}  

It is known that \eqref{eq:gradientplay} will converge locally asymptotically to a differential Nash equilibrium if the local linearization is a contraction~\cite{chasnov2019convergence}. 
Let 
\eqnn{g(x)=(D_1f_1(x),D_2f_2(x))
\label{eq:gameform}}
be the vector of individual gradients and let $Dg(x)$ be its Jacobian---i.e., the \emph{game Jacobian}. Further, let $\sigma_p(A)\subset\mb{C}$
denote the \emph{point spectrum} (or \emph{spectrum}) of the matrix $A$, and $\rho(A)$ its \emph{spectral radius}. 
Then, $x$ is \emph{locally exponentially stable} if and only if $\rho(I-\Gamma Dg(x))<1$, where $\Gamma=\mathrm{blockdiag}(\gamma_1 I_{\dimm_1},\gamma_2I_{\dimm_2})$
is a diagonal matrix and $I_{\dimm_i}$ is the identity matrix of dimension $\dimm_i$. 
The map $I-\Gamma Dg(x)$ is the local linearization of \eqref{eq:gradientplay}. 
Hence, to study stability (and, in turn, convergence) properties it is useful to analyze the spectrum of not only the map $I-\Gamma Dg(x)$ but also $Dg(x)$ itself. 

For instance, when $\gamma=\gamma_1=\gamma_2$, the spectral mapping theorem
tells us that $\rho(I-\gamma Dg(x))=\max_{\lambda\in \sigma_p(Dg(x))}|1-\gamma \lambda|$ so that understanding the spectrum of $Dg(x)$ is imperative for understanding convergence of the discrete time update. On the other hand, when $\gamma_1\neq \gamma_2$, we write the local linearization as $I-\gamma_1\Learnrate Dg(x)$ where $\Learnrate=\mathrm{blockdiag}(I_{\dimm_1},\tau I_{\dimm_2})$ and $\tau=\gamma_2/\gamma_1$ is the learning rate ratio. 
Again, via the spectral mapping theorem, when $I-\gamma_1\Learnrate Dg(x)$ is a contraction for different choices of learning rate $\gamma_1$ is determined by the spectrum of $\Learnrate Dg(x)$.
Hence, given a \emph{fixed point}  $\fixedpoint{x}$ (i.e., $g(\fixedpoint{x})=0$), we study the stability properties of the limiting continuous time dynamical system---i.e., $\dot{x}=-g(x) $ when $\gamma_1=\gamma_2$ and $\dot{x}=-\Learnrate g(x)$ otherwise. From here forward, we will simply refer to the system $\dot{x}=-\Learnrate g(x)$ and point out when $\Learnrate=I_{d_1+d_2}$ if not clear from context.

\paragraph{Partitioning the Game Jacobian} Let 
$\fixedpoint{x}=(\fixedpoint{x}_1,\fixedpoint{x}_2)$ 
be a joint action profile such that 
$g(\fixedpoint{x})=0$.
Towards better understanding the spectral properties of $Dg(\fixedpoint{x})$ (respectively, $\Learnrate Dg(\fixedpoint{x})$), we 
partition~$Dg(\fixedpoint{x})$ into blocks:
\eqnn{%
 J(x)=
  \bmat{%
  -D_1^2\cost_1(\fixedpoint{\xx}) & 
  -D_{12}\cost_1(\fixedpoint{\xx}) \\
  -D_{21}\cost_2(\fixedpoint{\xx}) & 
  -D_{2}^2\cost_2(\fixedpoint{\xx})
 }
 =
 \bmat{\Aa & \Bb \\
  \Cc & \Dd}.
 \label{eq:game-jacobian}
}
A differential Nash equilibrium (the second order conditions of which are sufficient for a local Nash equilibrium) is such that 
$\Aa<0$ and
$\Dd<0$. On the other hand, as noted above, $J$ is Hurwitz or stable if its point spectrum $\sigma_p(J)\subset \mb{C}_-^\circ$. 
Moreover,
since the diagonal blocks are symmetric,
$\J$ is 
similar to the matrix 
in Fig~\ref{fig:blockdiag2p}.
For the remainder of the paper, 
we will study the $Dg$ at a given
fixed point $\fixedpoint{x}$ as defined in~\eqref{eq:game-jacobian}.
\begin{figure}[ht]
  \begin{minipage}[c]{0.4\columnwidth}
    	$
 J(\fixedpoint{\x},\fixedpoint{\y})
 \sim \left[\begin{tikzpicture}[scale=0.5,baseline={([yshift=-.5ex]current bounding box.center)}]

		\draw [ultra thick](0,2) -- (.7,1.3);
		\draw [ultra thick] (.7,1.3) -- (2,0);
		\fill [color=gray,color=gray,draw=white,very thick] 
		(0,0) rectangle (.7,1.3);
		\fill [color=gray,draw=white,very thick] 
		(.7,1.3) rectangle (2,2);
	\end{tikzpicture}\right]
 $
  \end{minipage}\hfill
  \begin{minipage}[c]{0.59\columnwidth}
    \caption{\emph{Similarity}: the game Jacobian in~\eqref{eq:game-jacobian} 
    is similar to a matrix with diagonal block-diagonals.
    }
  \label{fig:blockdiag2p}
  \end{minipage}
\end{figure}

\paragraph{Classes of Games} Different classes of games can be characterized via $J$.   For instance,  a 
 \emph{zero-sum game}, where $\cost_1\equiv -\cost_2$, is such that $J_{12}=-J_{21}^\top$. On the other hand, %
 a game
 $\mc G = (\cost_1, \cost_2)$
 is a \emph{potential game} if and only if
 $D_{12}f_1 \equiv D_{21}f_2^\top$~\cite[Thm.~4.5]{monderer1996potential}, 
 which implies that $J_{12}=J_{21}^\top$.

\subsection{Spectrum of Block Matrices}
\label{sec:qnr}
One useful tool for characterizing the spectrum of a block operator matrix is the numerical range and quadratic numerical range, both of which 
contain the operator's spectrum~\cite{tretter2008spectral} and therefore all of its eigenvalues.
The \emph{numerical range} of $\J$ is defined by
\eqn{
   \NR(\J)=\{
\langle \J z,z\rangle:\ 
z\in \mb{C}^{\dimm_1+\dimm_2},\ %
\|z\|_2=1\}, 
}
and is convex. 
Given a block operator $J$,
let
\begin{equation}
\J_{v,w}=\bmat{\langle \J_{11}v,v\rangle & \langle\J_{12}w,v\rangle\\ \langle\J_{21}v,w\rangle & \langle\J_{22}w,w\rangle} 
\label{eq:jvw}
\end{equation}
where $v\in \mb{C}^{\dimm_1}$ and $w\in \mb{C}^{\dimm_2}$.
The \emph{quadratic numerical range} of $\J$,
defined by
\eqnn{
\NR^2(\J)=\bigcup_{v\in \mc{S}_1, w\in \mc{S}_2 }\sigma_p(\J_{v,w}),
\label{eq:qnr}}
is the union of the spectra of~\eqref{eq:jvw}
 where $\sigma_p(\cdot)$ denotes the (point) spectrum 
 of its argument 
 and $\mc{S}_i=\{z\in \mb{C}^{\dimm_i}:\ \|z\|_2=1\}$.
It is, in general, a non-convex subset of $\mb C$.
The quadratic numerical range~\eqref{eq:qnr}
is equivalent to the set of solutions of 
the characteristic polynomial
\eqnn{%
  \eig^2 
  &- \eig(\langle \Aa v,v\rangle  
  + \langle \Dd w,w\rangle ) 
  + \langle \Aa v,v\rangle\langle \Dd w,w\rangle  \\
  &- \langle \Bb v,w\rangle \langle \Cc w,v\rangle  = 0
  \label{eq:characteristic-polynomial}
}
for  $v\in \mc{S}_1$ and $w\in \mc{S}_2$. We use the notation $\langle \J x,y\rangle=\conjugate{x} \J y$ 
to denote the inner product. Note that $\NR^2(\J)$ is a subset of $\NR(\J)$ and, as previously noted, contains $\sigma_p(\J)$.
Albeit non-convex, $\NR^2(\J)$ provides a tighter characterization of the spectrum\footnote{There are numerous computational approaches for estimating the numerical ranges $\NR(\cdot)$ and $\NR^2(\cdot)$ (see, e.g., \cite[Sec. 6]{langer2001new}).}.

\begin{example}
Consider the game Jacobian of the zero-sum game $(f,-f)$ defined by 
cost $f:\R^2\times\R^2 \to \R$,
\[\cost(x,y)=-\tfrac{1}{2}x_1^2+\tfrac{5}{2}x_2^2+7y_1x_1-3y_2x_2-2y_1^2-6y_2^2.\]
The numerical range, quadratic numerical range, spectrum and diagonal entries of $\J$, defined using the origin as the fixed point,
are plotted in Fig.~\ref{fig:qnum-toy}. 
In this example, the origin is not 
a differential Nash equilibrium since $D^2_1f_1(0,0)$ is indefinite, yet it is an exponentially stable equilibrium of $\dot{x}=-g(x)$ since all the eigenvalues of $J$ are all negative.
\label{ex:qnum-toy}
\end{example}

Observing that the quadratic numerical range for a block $2\times 2$ matrix $J$ derived from a game on a finite dimensional Euclidean space reduces to characterizing the spectrum of $2\times 2$ matrices, 
we first characterize stability properties of scalar $2$-player continuous games.

\begin{figure}[t]
\centering
\begin{tikzpicture}
\begin{axis}[
width=.38\textwidth,
height=.34\textwidth,
thick,
axis on top,
axis x line=none,
axis y line = none,
ylabel={\small $\imag(\lambda)$},
xlabel={},
xtick=\empty,ytick=\empty,
xlabel style={below left},
ylabel style={below left},
clip=false]
\addplot graphics
[xmin=-1.5,xmax=1,ymin=-1.5,ymax=1.5]
{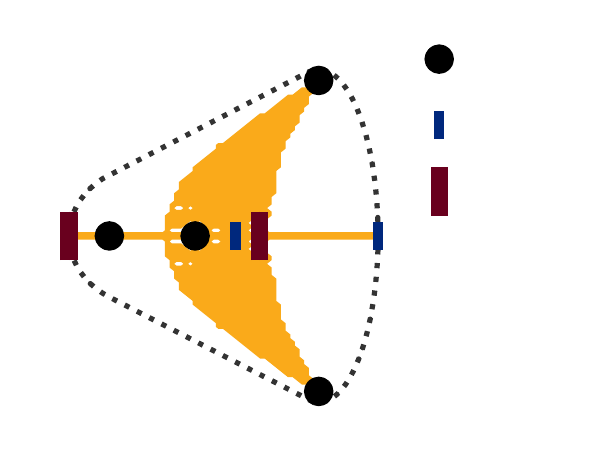};

\draw [<-] (axis cs:-0.5,-0.3) -- (axis cs:-.8,-.8) node [left  ] {$\QNR(\J)$}; 
\draw [-stealth] (axis cs:0,-1.3) -- (axis cs:0,1.3);
\draw (axis cs:.4, 1.1) node [right] {$\spec(J)$};
\draw (axis cs:.4, .68) node [right] {$\spec(J_{11})$};
\draw (axis cs:.4, .27) node [right] {$\spec(J_{22})$};
\draw [-stealth] (axis cs:0,0) -- (axis cs:1.5,0) node [below left] {\small $\real(\lambda)$};
\draw (axis cs:-1.5,0) -- (axis cs:-1.2,0);
\end{axis}
\end{tikzpicture}
\caption{
\emph{
Spectrum of a stable
equilibrium that is not Nash.}
The spectrum of $\J$, $\J_{11}$, and $\J_{22}$ in Example~\ref{ex:qnum-toy} are 
contained
in the numerical range (convex dashed region) and quadratic numerical range (non-convex region) of $J$.
The eigenvalues of $J$ are in the left plane, hence the fixed point is stable under gradient play~\eqref{eq:gradientplay}.
However, the first player's $J_{11}$ is indefinite, hence the fixed point is not a Nash equilibrium.
}
\label{fig:qnum-toy}
\end{figure}
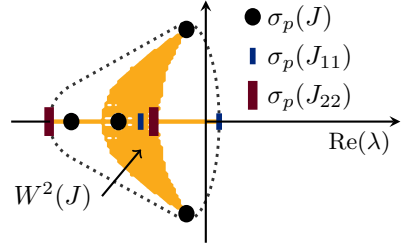

\section{Stability of $2$-player scalar games}
\label{sec:main}

We characterize the stability of differential Nash equilibria in 2-player scalar continuous games. 
Consider a game $(f_1,f_2)$
with action spaces
$X_1,X_2
\subseteq \R$.
Let $x$ be a fixed point of~\eqref{eq:gameform} such that $g(x)=0$.
We decompose its game Jacobian~\eqref{eq:game-jacobian} 
into components that reflect the dynamic interaction between the players.
\subsection{Jacobian Decomposition: $2\times2$ case}

Consider 
the decomposition of a $\R^{2\times 2}$ game Jacobian
\begin{equation}
 \label{eq:2x2:decomp}
J(x)=\bmat{a & b \\ c & d}  
= \bmat{\gG & -\kK \\ \kK & \gG} + \bmat{\fF & \pP \\ \pP & -\fF} 
\end{equation}
where $\gG=\frac{1}{2}(a+d)$, $\fF=\frac{1}{2}(a-d)$, $\pP =\frac{1}{2}(b+c)$, $\kK=\frac{1}{2}(c-b)$.
Let 
$\trace(J)$ be its trace,
$\det(J)$ be its determinant, and 
$\disc(J)$ be the discriminant of its characteristic polynomial%
.\footnote{
The characteristic polynomial 
of $J$ is 
$\lambda \mapsto \det(J-\lambda I)$ and its discriminant is $\trace(J)^2 - 4\det(J)$ for $J \in \R^{2\times 2}$.
}
Several directly verifiable 
quantities 
are stated. 
\begin{statement}
\label{main-statement}
Given a matrix $\J \in \mb R^{2\times 2}$
and its 
spectrum $\sigma_p(J)=\{\eig_1, \eig_2\}$, 
the above decomposition gives rise to the following conditions:  
\begin{align*}
\trace(\J) & = 
\lambda_1 + \lambda_2 = 
a+d
=
2
\gG, \\
\det\big(\J\big) & = 
\lambda_1\lambda_2 = ad-bc 
 = (\gG^2 + \kK^2) - (\fF^2 + \pP^2),
\\
\disc\big(\J\big) 
&= 
(\lambda_1 + \lambda_2)^2 - 4\lambda_1 \lambda_2
 = 
4(\fF^2 + \pP^2 - \kK^2), \\
\lambda_{1,2} & =
\tfrac{1}{2}\big(
 \trace(\J) \mp
 \sqrt{\disc(J)}
\big)
=
\gG \mp 
\sqrt{\fF^2 + \pP^2 - \kK^2}.
\end{align*}
\end{statement}

 The change of coordinates from $(a,b,c,d)$ to $(m,h,p,z)$
in Statement~\ref{main-statement} 
provides 
 important insights into linear vector fields and, in particular, to games.
The stability of vector field $\dot x = Jx$ is given by the trace and determinant conditions.
\begin{proposition} %
\label{prop:stability}
The matrix $\J \in \mathbb{R}^{2 \times 2}$ is stable if and only if $\gG^2 + \kK^2 > \fF^2 + \pP^2$ and
$\gG < 0$.
\end{proposition}
\begin{proof}
Statement \ref{main-statement} and direct computation show that these conditions are equivalent to $\lambda_1 + \lambda_2 < 0$ and $\lambda_1\lambda_2 >0$, 
well-known conditions for stability of $2\times 2$ systems (illustrated in Fig.~\ref{fig:tracedet:tracedet}).
\end{proof}

\subsection{Discussion of Decomposition}
\label{sec:insights}

The purpose of the decomposition into the alternative coordinates
is to geometrically---and thus
more directly---assess the conditions for stability of a differential Nash equilibrium. 
\ifsixpages
In particular,
the conditions for a fixed point of the game dynamics to be a differential Nash equilibrium
are $a<0$ and $d < 0$,
or equivalently,
$\gG < -|\fF|$,
which are represented 
by the left-shaded regions in Fig.~\ref{fig:gkfp:labels}--\ref{fig:gkfp:stable&nash}.
Moreover,
the conditions for a fixed point of the game dynamics to be stable (i.e., for $\sigma_p(J)\subset \mb{C}_-^\circ$) are $m<0$ and 
$m^2 + z^2 > h^2 + p^2$,
which are visible only
when using the decomposition
in Fig.~\ref{fig:gkfp:stable&nash}.
\else

\paragraph{Relationship to complex plane}
Fig.~\ref{fig:eigs} 
plots the coordinates of
$\gG,\kK,\fF,\pP$
relative to each other
to illustrate
the decomposition in
Statement~\ref{main-statement}.
If 
$\fF=0, \pP=0$, 
then the eigenvalues of $\J$ are $\lambda_{1,2} = \gG \mp \kK i$. Fig.~\ref{fig:eigs:complex} corresponds to a plot of eigenvalues in the complex plane.
Stability is given by the familiar open-left half plane condition: $\spec(\J) \subset \mathbb{C}_-^\circ$.  
If $\fF\neq 0$ or $\pP \neq 0$ 
a circular region in the center of the plane
expands the values of $\gG,\kK$ 
for which the eigenvalues of the matrix are purely real.
Fig.~\ref{fig:eigs:gkfp} shows
that
the eigenvalues are purely real if and only if
$\kK^2 \leq \fF^2 + \pP^2$.

\paragraph{Effect of rotation in game vector fields}

Note the similarity between \eqref{eq:2x2:decomp} and the well-known symmetric/skew-symmetric (Helmholtz) decomposition

\begin{equation}
 \label{eq:2x2:helmholtz}
J(x)=\bmat{a & b \\ c & d}  
= \bmat{\gG + h & p\\ p & m-h} + \bmat{0 & -z \\ z & 0}. 
\end{equation}
Assuming that $m < 0$, from Proposition \ref{prop:stability} we can see that increasing 
the rotational component of the Jacobian helps stability.  Increasing the relative magnitude of $p$, the non-rotational interaction term hurts stability.  If there is no rotational component, ie. $J$ is symmetric, $p$'s negative impact on stability can be seen directly from the Schur complement\footnote{%
The Schur complements of the matrix in~\eqref{eq:game-jacobian} 
are $J_{11}-J_{12}J_{22}^{-1}J_{21}$ (where $J_{22}$ is invertible) and $J_{22}-J_{21}J_{11}^{-1}J_{12}$ (where $J_{11}$ is invertible).}. In this case $J$ is stable iff $J < 0$ and thus stability requires that both the diagonals and the Schur complement are negative: $a<0$, $d<0$, and $a-p^2d^{-1}<0$.  If $d<0$, increasing $p$ can only increase the Schur complement.

 \begin{figure}[t]
 \centering
 \subfloat[][
 The complex plane.
 ]{
  \begin{tikzpicture}[scale=0.35]
   \coordinate (gk) at (-2,1.5);
   \coordinate (fp) at (0,0);
   \coordinate (lim) at (4,4);
   \fill [black] (0,0) circle (2pt);

   \begin{scope}[transparency group]
   \begin{scope}[blend mode=multiply]
 
    \fill[stable fill] let \p1 = (fp), \n{radius} = {veclen(\x1, \y1)}, \p2 = (lim) in
     (0, -\y2) -- (-\x2,-\y2) -- (-\x2, \y2) -- (0, \y2) -- (0, \n{radius})
     arc [start angle=-90, end angle=90, radius=-\n{radius}] -- cycle;
    \draw[stable draw] 
     let \p1 = (fp), \n{radius} = {veclen(\x1, \y1)}, \p2 = (lim) in
     (0, \y2) -- (0, \n{radius}) arc [start angle=-90, end angle=90, radius=-\n{radius}] -- (0,-\y2);
    \draw[dotted] 
     let \p1 = (fp), \n{radius} = {veclen(\x1, \y1)} in
     (0, \n{radius}) arc [start angle=90, end angle=-90, radius=\n{radius}];
    \draw[real draw]
     let \p1 = (fp), \n{radius} = {veclen(\x1, \y1)}, \p2 = (lim) in
     (\x2, -\n{radius}) -- (-\x2, -\n{radius});
    \draw[real draw]
     let \p1 = (fp), \n{radius} = {veclen(\x1, \y1)}, \p2 = (lim) in
     (\x2, \n{radius}) -- (-\x2, \n{radius});
    \fill[real fill]
     let \p1 = (fp), \n{radius} = {veclen(\x1, \y1)}, \p2 = (lim) in
     (\x2, -\n{radius}) rectangle (-\x2, \n{radius});
   \end{scope}
   \end{scope}

   \draw [-stealth,thick] (0,0) -- (gk) node [above] {$\gG+\kK i$};
   \draw [-stealth,thick] let \p1=(gk) in (0,0) -- (\x1, -\y1) node [below] {$\gG-\kK i$};

   \draw (-2,-4) node [above] {Stable};
  \end{tikzpicture}
  \label{fig:eigs:complex}
 }
 \subfloat[][
 A representation of the $m,z,h,p$ coordinates. 
 ]{
  \begin{tikzpicture}[scale=0.35]
   \coordinate (gk) at (-1,2);
   \coordinate (fp) at (1,-1);
   \coordinate (lim) at (4,4);
   \fill [black] (0,0) circle (2pt);

   \begin{scope}[transparency group]
   \begin{scope}[blend mode=multiply]
 
    \fill[stable fill] let \p1 = (fp), \n{radius} = {veclen(\x1, \y1)}, \p2 = (lim) in
     (0, -\y2) -- (-\x2,-\y2) -- (-\x2, \y2) -- (0, \y2) -- (0, \n{radius})
     arc [start angle=-90, end angle=90, radius=-\n{radius}] -- cycle;
    \draw[stable draw] 
     let \p1 = (fp), \n{radius} = {veclen(\x1, \y1)}, \p2 = (lim) in
     (0, \y2) -- (0, \n{radius}) arc [start angle=-90, end angle=90, radius=-\n{radius}] -- (0,-\y2);
    \draw[dotted, thick, color=gray] 
     let \p1 = (fp), \n{radius} = {veclen(\x1, \y1)} in
     (0, \n{radius}) arc [start angle=90, end angle=-90, radius=\n{radius}];

    \draw[real draw]
     let \p1 = (fp), \n{radius} = {veclen(\x1, \y1)}, \p2 = (lim) in
     (\x2, -\n{radius}) -- (-\x2, -\n{radius});
    \draw[real draw]
     let \p1 = (fp), \n{radius} = {veclen(\x1, \y1)}, \p2 = (lim) in
     (\x2, \n{radius}) -- (-\x2, \n{radius});
    \fill[real fill]
     let \p1 = (fp), \n{radius} = {veclen(\x1, \y1)}, \p2 = (lim) in
     (\x2, -\n{radius}) rectangle (-\x2, \n{radius});
  
   \end{scope}
   \end{scope}

   \draw [-stealth,thick] (0,0) -- (fp) node [right] {$(\fF,\pP)$};
   \draw [-stealth,thick] (0,0) -- (gk) node [above left] {$(\gG,\kK)$};
   \draw (-2,-4) node [above] {Stable};
   \draw (4,0) node 
   [left] {\small Real};
   \draw (4,2.7) node 
   [left]
   {\small Imaginary};
   \draw (4,-2.7) node 
   [left]
   {\small Imaginary};
  \end{tikzpicture}
  \label{fig:eigs:gkfp}
 }
 \caption{%
 \emph{Visualization of Statement~\ref{main-statement}:}
 If $\fF$ and $\pP$ are zero, then the eigenvalues 
 of $\J$
 are 
 $\eig_{1,2}=\gG \mp \kK i$.
 If $\fF$ and/or $\pP$ are non-zero, 
 then a circle 
 centered around
 the origin
 with radius $\sqrt{h^2+p^2}$ is excluded from left-half stability region.}
 
 \label{fig:eigs}
\end{figure}
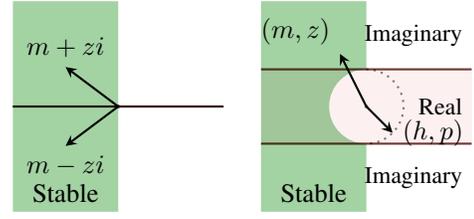

\begin{figure}[t]
\centering
 \subfloat[width=0.5\columnwidth][
Level sets of $\det(\J)=\eig_1 \eig_2$.]{
  \begin{tikzpicture}[scale=0.7,domain=-2:2]
  \clip (-2,-2) rectangle (2,2);
  
  \fill [gray!10!white] (-2,-2) -- (2,-2)
  -- (-2,2) -- cycle;

  \fill [stable fill] (-2,-2) rectangle (0,0);

  \draw[-stealth,thick] (-2,0) -- (2,0) node[below left]{$|\eig_1|$};
  \draw[-stealth,thick] (0,-2) -- (0,2) node[below right]{$|\eig_2|$};
  
  \foreach \l / \k in {.4/.2,.2/.4,.1/.6,.05/.8}
  \draw[domain=0.01:2,samples=128,range=-2:2,opacity=\k] plot (\x,\l/\x);

  \foreach \l / \k in {.4/.2,.2/.4,.1/.6,.05/.8}
  \draw[domain=0.01:2,samples=128,range=-2:2,opacity=\k] plot (-\x,-\l/\x);

  \draw (-1.05,-2) node [text width=1.2cm, above] {\footnotesize 
  $\det(\J)$ \\ $>0$,\\ 
  $\sq{\trace(\J)<0}$
  };
  \draw (-1,1) node [rotate=-45] {\footnotesize $\trace(\J) = 0$};
  
  \draw [very thick,domain=-.2:2] plot (\x,-\x);
  \draw [very thick,domain=-2:-1.8] plot (\x,-\x);

  \end{tikzpicture}
  \label{fig:tracedet:eigs}
 }
 \subfloat[width=0.5\columnwidth][
 Real or imaginary eigenvalues.
]{
  \begin{tikzpicture}[scale=0.7,domain=-2:2,samples=100]
  \clip (-2,-2) rectangle (2,2);
  \fill [stable fill] 
  (0,0) rectangle (-2,2);

  \fill [real fill] (-2,-2)
  -- plot (\x,{\x*\x/2.5})
  -- (2,-2) -- cycle;
  \draw (-1.1,.72) node [rotate=-40] 
  {\footnotesize $\textrm{disc}{(J) = 0}$};
 
  \draw[-stealth,thick] (-2,0) -- (2,0) node[below left]{$\trace(\J)$};
  \draw[-stealth,thick] (0,-1) -- (0,2) node[below right]{$\det(\J)$};
    \draw [very thick,domain=-.2:2] plot (\x,{\x*\x/2.5});
  \draw (0,-1.5) node {$\eig_1,\eig_2\in\mb R$};
  \end{tikzpicture}
  \label{fig:tracedet:tracedet}
}
\caption{\emph{Visualization of~\autoref{prop:stability}}: $\dot y = \J(x) y$ is \tightcolorbox{mygreen!80!white}{stable}
 $\Longleftrightarrow\ \det(J) > 0$ and $\trace(J) < 0$.
}
\label{fig:tracedet}
\end{figure}

\subsection{Types of Games}

The decomposition also provides a natural classification of 2-player scalar games into four types based on specific coordinates being zero, as illustrated in Fig.~\ref{fig:gametype}.

\paragraph{Potential games ($z=0$)}
    The point $(\gG, \kK)$ lives on the horizontal axis 
   in Fig.~\ref{fig:gametype:k0},
   thus stable fixed points are a subset of Nash equilibria.  Since $z = 0$, Proposition \ref{prop:stability} indicates that increasing $p$, the interaction term between the players, and increasing $h$, the difference in curvature between the two players both only hurt stability.

\paragraph{Zero-sum games ($p=0$)}
   The point $(\fF , \pP)$ lives on the horizontal axis
   in 
   Fig.~\ref{fig:gametype:p0},
   thus all Nash equilibria are stable, but not all stable fixed points are Nash.
   The magnitude of the interaction term $\kK$ helps stability and may make a fixed point stable even if it is not Nash.  
   \ifsixpages
   Intuitively, a strong enough interaction term can cause a player which is at its action under the Nash equilibrium with stronger negative curvature to pull another player with weaker positive curvature toward a fixed point even if that point is a local 
   maximum %
   for the weaker player. 
   \fi

\paragraph{Hamiltonian games ($m=0$)} 
    The point $(\gG, \kK)$ lives on the vertical axis 
    in Fig.~\ref{fig:gametype:g0},
    thus no strict Nash equilibria can exist.  At best these games are marginally stable if $|\kK|$ is large enough relative to the magnitude of $(\fF , \pP)$.

\begin{figure}[h]
\centering
 \subfloat[width=0.49\columnwidth][
 Geometry of decomposition in~\eqref{eq:2x2:decomp}. 
 ]{
  \begin{tikzpicture}[scale=0.40]
   \coordinate (ab) at (-1,-3);
   \coordinate (dc) at (-3,1);
   \coordinate (lim) at (4,4);

   \fill [black] (0,0) circle (2pt);

   \draw [gray] let \p2=(lim) in (-\x2,0) -- ($(0,0)$);
   \draw [gray,-stealth] let \p2=(lim) in (0,0) -- ($(\x2,0)$);
   \draw [gray,-stealth] let \p2=(lim) in (0,-\y2) -- (0,\y2);

   \fill[nash fill] let \p1=(lim) in    
    (-\x1,-\y1) rectangle (0, \y1);
   \draw[nash draw] let \p1=(lim) in    
    (0,-\y1) -- (0, \y1);
    
   \draw [-stealth,thick] (0,0) -- (ab) node [left] {$(a,b)$};
   \draw [-stealth,thick] (0,0) -- (dc) node [above] {$(d,c)$};

   \draw [dotted,thick] let \p1=(dc) in
    (\x1,0) -- (dc);
   \draw [dotted,thick] let \p1=(ab) in
    (\x1,0) -- (\p1);
   \draw let \p1=(ab), \p2=(dc), \n1={\x1/2+\x2/2} in
    ($(\n1, 0)$) node [below, xshift=-.5ex ] {$\gG$};
   \fill let \p1=(ab), \p2=(dc), \n1={\x1/2+\x2/2} in
    ($(\n1, 0)$);%

   \draw  let \p1=(ab), \p2=(dc), \n1={\y1/2+\y2/2} in
    ($(0, \n1)$) node [right] {$\pP$};
   \fill let \p1=(ab), \p2=(dc), \n1={\y1/2+\y2/2} in
    ($(0, \n1)$);%

   \draw [-stealth,dashed,thick] let \p1=(ab), \p2=(dc), 
    \n1={\x1/2+\x2/2}, %
    \n2={\x1/2-\x2/2}, %
    \n3={\y1/2+\y2/2}, %
    \n4={\y2/2-\y1/2}  %
    in
    ($(\n1,0)$) --
    ($(\n1,-.15)$) -- node [below] {\small $\fF$}
    ($(\n1+\n2,-.15)$) ;

   \draw [-stealth,dashed,thick] let \p1=(ab), \p2=(dc), 
    \n1={\x1/2+\x2/2}, %
    \n2={\x1/2-\x2/2}, %
    \n3={\y1/2+\y2/2}, %
    \n4={\y2/2-\y1/2}  %
    in
    ($(0, \n3)$) --
    ($(.15, \n3)$) -- node [right] {\small $\kK$}
    ($(.15, \n3+\n4)$);

   \draw [thin,gray,-stealth] let \p2=(lim) in (0,{\y2-0.01}) -- (0,\y2);
   \draw [dotted,thick] let \p1=(dc) in
    (0,\y1) -- (dc);
   \draw [dotted,thick] let \p1=(ab) in
    (0,\y1) -- (\p1);
   \draw (-2,4) node [below] {\small Nash};
    \draw (0,4) node [gray,below right] {\small b,c};
    \draw (4,0) node [gray,below left] {\small a,d};
  \end{tikzpicture}
  \label{fig:gkfp:labels}
 }
 \subfloat[width=0.49\columnwidth][
 Change of coordinates reveals regions of stability.
 ]{
  \begin{tikzpicture}[scale=0.40]
   \coordinate (gk) at (-2,2.5);
   \coordinate (fp) at (1,-1);
   \coordinate (lim) at (4,4);
   \fill [black] (0,0) circle (2pt);

   \draw [thin,gray,-stealth] let \p2=(lim) in (-\x2,0) -- ($(\x2,0)$);
   \draw [thin,gray,-stealth] let \p2=(lim) in (0,-\y2) -- (0,\y2);

   \fill[stable fill]
    let \p1 = (fp),
    \n{radius} = {veclen(\x1, \y1)},
    \p2 = (lim)
    in
    (0, -\y2) -- (-\x2,-\y2) -- (-\x2, \y2) --
    (0, \y2) -- (0, \n{radius})
    arc [start angle=-90, end angle=90, radius=-\n{radius}] 
    -- cycle;
   \draw[stable draw] 
    let \p1 = (fp), \n{radius} = {veclen(\x1, \y1)}, \p2 = (lim) in
    (0, \y2) -- (0, \n{radius})
    arc [start angle=-90, end angle=90, radius=-\n{radius}]
    -- (0,-\y2);
   \draw[thick, dotted] 
    let \p1 = (fp), \n{radius} = {veclen(\x1, \y1)} in
    (0, \n{radius}) arc [start angle=90, end angle=-90, radius=\n{radius}];

   \fill[nash fill] 
    let \p1 = (fp), \p2 = (lim)
    in
    (-\x2, \y2) -- ($(-{abs(\x1)}, \y2)$) -- 
    ($(-{abs(\x1)}, -\y2)$) -- (-\x2, -\y2) -- cycle;
   \draw[nash draw]
    let \p1 = (fp), \p2 = (lim)
    in ($(-{abs(\x1)},-\y2)$) -- ($(-{abs(\x1)},\y2)$);

   \draw [<-] (-.5,2) -- (1,2) node [right] {\small Stable};
   \draw[->] (-1.7,0.3) node [left] {\small Nash} -- (-1.2,0.3);
   \draw [-stealth,thick] (0,0) -- 
   (gk)  node [above] {$(\gG,\kK)$};
   \draw [-stealth,dashed,thick] let \p1 = (fp) in
    (0,0) -- ($(-{abs(\x1)}, \y1)$)
    node [below left] {$(\text{--}|\fF|, \pP)$};
    \draw (0,4) node [gray,below right] {\small z,p};
    \draw (4,0) node [gray,below left] {\small m,h};
   \draw [thin,gray,-stealth] let \p2=(lim) in (0,{\y2-0.01}) -- (0,\y2);
  \end{tikzpicture}
  \label{fig:gkfp:stable&nash}
}
\caption{%
\emph{Decomposition of a general scalar game.}
The rows vectors of $J$ are plotted in (a)
and the same matrix 
with a change of coordinates
is plotted in (b).
	\tightcolorbox{myblue!50!white}{Nash} regions 
	($m<-|h|$) and 
	\tightcolorbox{mygreen!80!white}{stability}
	regions ($m<0,m^2+z^2>h^2+p^2$) 
	are visible.
	Their set differences
	characterize the conditions for a stable non-Nash and unstable Nash equilibria.
}
\label{fig:gkfp}

\centering
 \subfloat[width=0.24\columnwidth][
 Potential:

 Stable $\subset$ Nash.
 ]{
  \begin{tikzpicture}[scale=0.2]
   \coordinate (gk) at (-3,0);
   \coordinate (fp) at (-1.414,1.414);
   \coordinate (lim) at (4,4);
   \fill [black] (0,0) circle (5pt);
    \fill[stable fill]
     let \p1 = (fp), \n{radius} = {veclen(\x1, \y1)}, \p2 = (lim) in
     (0, -\y2) -- (-\x2,-\y2) -- (-\x2, \y2) -- (0, \y2) -- (0, \n{radius})
     arc [start angle=-90, end angle=90, radius=-\n{radius}] -- cycle;
    \draw[stable draw] 
     let \p1 = (fp), \n{radius} = {veclen(\x1, \y1)}, \p2 = (lim) in
     (0, \y2) -- (0, \n{radius})
     arc [start angle=-90, end angle=90, radius=-\n{radius}] -- (0,-\y2);
    \draw[thick, dotted, color=gray] 
     let \p1 = (fp), \n{radius} = {veclen(\x1, \y1)} in
     (0, \n{radius}) arc [start angle=90, end angle=-90, radius=\n{radius}];
    \fill[nash fill] let \p1 = (fp), \p2 = (lim) in 
     (-\x2, \y2) -- ($(-{abs(\x1)}, \y2)$) -- ($(-{abs(\x1)}, -\y2)$) -- (-\x2, -\y2) -- cycle;
    \draw[nash draw] let \p1 = (fp), \p2 = (lim) in 
     ($(-{abs(\x1)},-\y2)$) -- ($(-{abs(\x1)},\y2)$);
   \draw [-stealth,very thick] (0,0) -- (gk);
   \draw [-stealth,very thick] let \p1=(gk) in (0,0) -- (-\x1, -\y1);
   \draw [-stealth,dashed, thick] (0,0) -- (fp);
   \draw [black] (2.2,3.5) node {$\kK=0$};
  \end{tikzpicture}
  \label{fig:gametype:k0}
}
 \subfloat[width=0.24\columnwidth][
Zero-sum:

Stable $\supset$ Nash
 ]{
  \begin{tikzpicture}[scale=0.2]
   \coordinate (gk) at (-3,3);
   \coordinate (fp) at (-2,0);
   \coordinate (lim) at (4,4);
   \fill [black] (0,0) circle (5pt);
    \fill[stable fill]
     let \p1 = (fp), \n{radius} = {veclen(\x1, \y1)}, \p2 = (lim) in
     (0, -\y2) -- (-\x2,-\y2) -- (-\x2, \y2) -- (0, \y2) -- (0, \n{radius})
     arc [start angle=-90, end angle=90, radius=-\n{radius}] -- cycle;
    \draw[stable draw] 
     let \p1 = (fp), \n{radius} = {veclen(\x1, \y1)}, \p2 = (lim) in
     (0, \y2) -- (0, \n{radius})
     arc [start angle=-90, end angle=90, radius=-\n{radius}] -- (0,-\y2);
    \draw[thick, dotted, color=gray] 
     let \p1 = (fp), \n{radius} = {veclen(\x1, \y1)} in
     (0, \n{radius}) 
     arc [start angle=90, end angle=-90, radius=\n{radius}];
    \fill[nash fill] let \p1 = (fp), \p2 = (lim) in 
     (-\x2, \y2) -- ($(-{abs(\x1)}, \y2)$) -- ($(-{abs(\x1)}, -\y2)$) -- (-\x2, -\y2) -- cycle;
    \draw[nash draw] let \p1 = (fp), \p2 = (lim) in 
     ($(-{abs(\x1)},-\y2)$) -- ($(-{abs(\x1)},\y2)$);
   \draw [-stealth,thick] (0,0) -- (gk);
   \draw [-stealth,dashed,very thick] (0,0) -- (fp);
   \draw [black] (2.2,3.5) node {$\pP=0$};
  \end{tikzpicture}
  \label{fig:gametype:p0}
}
\subfloat[width=0.24\columnwidth][
 Hamiltonian: 

 marginally stable at best.
 ]{
  \begin{tikzpicture}[scale=0.2]
   \coordinate (gk) at (0,3);
   \coordinate (fp) at (-1.414,1.414);
   \coordinate (lim) at (4,4);
   \fill [black] (0,0) circle (5pt);
    \fill[stable fill]
     let \p1 = (fp), \n{radius} = {veclen(\x1, \y1)}, \p2 = (lim) in
     (0, -\y2) -- (-\x2,-\y2) -- (-\x2, \y2) -- (0, \y2) -- (0, \n{radius})
     arc [start angle=-90, end angle=90, radius=-\n{radius}] -- cycle;
    \draw[stable draw] 
     let \p1 = (fp), \n{radius} = {veclen(\x1, \y1)}, \p2 = (lim) in
     (0, \y2) -- (0, \n{radius})
     arc [start angle=-90, end angle=90, radius=-\n{radius}] -- (0,-\y2);
    \draw[thick, dotted, color=gray] 
     let \p1 = (fp), \n{radius} = {veclen(\x1, \y1)} in
     (0, \n{radius}) arc [start angle=90, end angle=-90, radius=\n{radius}];
    \fill[nash fill] let \p1 = (fp), \p2 = (lim) in 
     (-\x2, \y2) -- ($(-{abs(\x1)}, \y2)$) -- ($(-{abs(\x1)}, -\y2)$) -- (-\x2, -\y2) -- cycle;
    \draw[nash draw] let \p1 = (fp), \p2 = (lim) in 
     ($(-{abs(\x1)},-\y2)$) -- ($(-{abs(\x1)},\y2)$);
   \draw [-stealth, very thick] (0,0) -- (gk);
   \draw [-stealth, very thick] let \p1=(gk) in (0,0) -- (-\x1, -\y1);
   \draw [-stealth,dashed, thick] (0,0) -- (fp);
   
   \draw [black] (2.2,3.5) node {\small $\gG=0$};
  \end{tikzpicture}
  \label{fig:gametype:g0}
}
 \subfloat[width=0.24\columnwidth][
 Matching:

 Stable $\subset$ Nash.
 ]{
  \begin{tikzpicture}[scale=0.2]
   \coordinate (gk) at (-3,3);
   \coordinate (fp) at (0,-2);
   \coordinate (lim) at (4,4);
   \fill [black] (0,0) circle (5pt);
    \fill[stable fill]
     let \p1 = (fp), \n{radius} = {veclen(\x1, \y1)}, \p2 = (lim) in
     (0, -\y2) -- (-\x2,-\y2) -- (-\x2, \y2) -- (0, \y2) -- (0, \n{radius})
     arc [start angle=-90, end angle=90, radius=-\n{radius}] -- cycle;
    \draw[stable draw] 
     let \p1 = (fp), \n{radius} = {veclen(\x1, \y1)}, \p2 = (lim) in
     (0, \y2) -- (0, \n{radius})
     arc [start angle=-90, end angle=90, radius=-\n{radius}] -- (0,-\y2);
    \draw[thick, dotted, color=gray] 
     let \p1 = (fp), \n{radius} = {veclen(\x1, \y1)} in
     (0, \n{radius}) arc [start angle=90, end angle=-90, radius=\n{radius}];
    \fill[nash fill] let \p1 = (fp), \p2 = (lim) in 
     (-\x2, \y2) -- ($(-{abs(\x1)}, \y2)$) -- ($(-{abs(\x1)}, -\y2)$) -- (-\x2, -\y2) -- cycle;
    \draw[nash draw] let \p1 = (fp), \p2 = (lim) in 
     ($(-{abs(\x1)},-\y2)$) -- ($(-{abs(\x1)},\y2)$);
   \draw [-stealth,thick] (0,0) -- (gk);
   \draw [-stealth,dashed,very thick] (0,0) -- (fp);
   \draw [-stealth,dashed,very thick] let \p1 = (fp) in (0,0) -- (\x1, -\y1);
   \draw [black] (2.2,3.5) node {\small $\fF=0$};
  \end{tikzpicture}
  \label{fig:gametype:f0}
}
\caption{
\emph{Stability and Nash
for different
classes of games.}
(a) Potential games:
symmetric interaction term only hurts stability.
(b) Zero-sum games: rotation can compensate for unhappy player.
(c) Hamiltonian games: players have
zero total curvature, $a+d=0$. 
(d) Matching curvature, $a=d$: there are no stable non-Nash equilibria.
}
\label{fig:gametype}

 \centering
 \subfloat[width=0.48\columnwidth][
 $\learnrate_1 > \learnrate_2$
 ]{
  \begin{tikzpicture}[scale=0.25]
   \coordinate (gk) at (0,2);
   \coordinate (fp) at (1,-1);
   \coordinate (lim) at (4,4);
   \fill [black] (0,0) circle (2pt);

   \draw[dotted, thick] 
    let \p1 = (fp), \n{radius} = {veclen(\x1, \y1)} in
    (-0.5, 1.32288)
    arc [start angle=110.705, end angle=-110.705, radius=\n{radius}];
   \fill[stable fill]
    let \p1 = (fp),
    \n{radius} = {veclen(\x1, \y1)},
    \n{radius2} = {veclen(\x1, \y1)-0.02},
    \p2 = (lim),
    \p3 = (fp)
    in
    (-0.5, -\y2) -- (-\x2,-\y2) -- (-\x2, \y2) --
    (-0.5, \y2) -- (-0.5, 1.32288)
    arc [start angle=-69.29, end angle=69.29, radius=-\n{radius}]
    -- cycle;
   \draw[stable draw] 
    let \p1 = (fp), \n{radius} = {veclen(\x1, \y1)}, \p2 = (lim) in
    (-0.5, \y2) -- (-0.5, 1.32288)
    arc [start angle=-69.29, end angle=69.29, radius=-\n{radius}]
    -- (-0.5,-\y2+0.02);
    
   \fill[nash fill] 
    let \p1 = (fp), \p2 = (lim)
    in
    (-\x2, \y2) -- ($(-{abs(\x1)}, \y2)$) -- 
    ($(-{abs(\x1)}, -\y2)$) -- (-\x2, -\y2) -- cycle;
   \draw[nash draw]
    let \p1 = (fp), \p2 = (lim)
    in ($(-{abs(\x1)},-\y2)$) -- ($(-{abs(\x1)},\y2)$);
    
    \draw[dotted,thick]
    let \p1 = (fp), \p2 = (lim)
    in ($(\x1,-\y2)$) -- ($(\x1,\y2)$);

   \draw [-stealth,thick] (0,0) -- 
   (gk) node [above] {$(\gG,\kK)$};
   \draw [-stealth,thick] (0,0) -- 
   (fp) node [below right] {$(\fF,\pP)$};
   \draw [-stealth,dashed,thick] let \p1 = (fp) in
    (0,0) -- ($(-{abs(\x1)}, \y1)$);
    
  \end{tikzpicture}
}
 \subfloat[width=0.48\columnwidth][
 $\learnrate_1 < \learnrate_2$
 ]{
  \begin{tikzpicture}[scale=0.25]
   \coordinate (gk) at (0,2);
   \coordinate (fp) at (1,-1);
   \coordinate (lim) at (4,4);
   \fill [black] (0,0) circle (2pt);

   \draw[dotted, thick] 
    let \p1 = (fp), \n{radius} = {veclen(\x1, \y1)} in
    (0.5, 1.32288)
    arc [start angle=69.29, end angle=-69.29, radius=\n{radius}]
    (0.5, -1.32288);
   \fill[stable fill]
    let \p1 = (fp),
    \n{radius} = {veclen(\x1, \y1)},
    \p2 = (lim)
    in
    (0.5, -\y2) -- (-\x2,-\y2) -- (-\x2, \y2) --
    (0.5, \y2) -- (0.5, 1.32288)
    arc [start angle=-110.705, end angle=110.705, radius=-\n{radius}]
    -- cycle;
   \draw[stable draw] 
    let \p1 = (fp), \n{radius} = {veclen(\x1, \y1)}, \p2 = (lim) in
    (0.5, \y2) -- (0.5, 1.32288)
    arc [start angle=-110.705, end angle=110.705, radius=-\n{radius}]
    -- (0.5,-\y2);

   \fill[nash fill] 
    let \p1 = (fp), \p2 = (lim)
    in
    (-\x2, \y2) -- ($(-{abs(\x1)}, \y2)$) -- 
    ($(-{abs(\x1)}, -\y2)$) -- (-\x2, -\y2) -- cycle;
   \draw[nash draw]
    let \p1 = (fp), \p2 = (lim)
    in ($(-{abs(\x1)},-\y2)$) -- ($(-{abs(\x1)},\y2)$);
    
    \draw[dotted,thick]
    let \p1 = (fp), \p2 = (lim)
    in ($(\x1,-\y2)$) -- ($(\x1,\y2)$);

   \draw [-stealth,thick] (0,0) -- 
   (gk) node [above] {$(\gG,\kK)$};
   \draw [-stealth,thick] (0,0) -- 
   (fp) node [below right] {$(\fF,\pP)$};
   \draw [-stealth,dashed,thick] let \p1 = (fp) in
    (0,0) -- ($(-{abs(\x1)}, \y1)$);
    
  \end{tikzpicture}
}
 \caption{%
 \emph{Time-scale separation affects stability.}
 The learning rate ratio $\tau=\learnrate_2/\learnrate_1>0$ affects the stability of the game dynamics.
 The factor
 $\beta = \tfrac{\tau - 1}{\tau + 1}$
 expands or shrinks the region for stability.  The condition $m < 0$ becomes $m < \beta h$. Note that $-1 \leq \beta \leq 1$ for $\tau \geq 0$. 
 \ifsixpages
 For $\beta \to \pm 1$,
 the region's vertical boundary 
 approaches $\pm h$.
 \fi
 }
 \label{fig:rates}
 \end{figure}

\paragraph{Matching-curvature games, ($h=0$)}
The point $(\fF , \pP)$ lives on the vertical axis 
in Fig.~\ref{fig:gametype:f0}, 
so any stable point is also a Nash. 
Any fixed point with $a,d$ having the same sign can be rescaled to have matching curvature $\learnrate_1 a=\learnrate_2 d$ by a choice of non-uniform learning rates $\learnrate_1,\learnrate_2>0$.

\section{Certificates for Stability of Game Dynamics}
\label{sec:certificates}
\subsection{Stability: Uniform Learning Rates}
\label{sec:learningrates}

For a game $\mc{G}=(\cost_1,\cost_2)$, let the set of differential Nash equilibria be denoted ${\tt DNE}(\mc{G})$ and let the stable points of $\dot{x}=-g(x)$ be ${\tt S}(\mc{G})$. Let $\overline{\tt DNE}(\mc{G})$ and $\overline{\tt S}(\mc{G})$ be their respective complements.
The intersections of these sets characterize the stability/instability of Nash/non-Nash equilibria.
\begin{theorem}[Certificates for $2$-Player Scalar Games]
\label{prop:classification}
Consider a 
game $\mc{G}=(\cost_1, \cost_2)$ on $X_1\times X_2\subseteq \R^2$. Let $\fixedpoint{x}$ be a fixed point of~\eqref{eq:gameform} and let $m,h,p,z$ be defined by~\eqref{eq:2x2:decomp}. %
The following equivalences hold:
\begin{enumerate}
[label=(\roman*)]
    \item $x\in {\tt DNE}(\mc{G})\cap {\tt S}(\mc{G}) \Longleftrightarrow$\\ $\qquad\quad\{
\gG < -|\fF|\}\wedge
\{
\gG^2+\kK^2 > \fF^2 + \pP^2\}
$.
\item $x\in {\tt DNE}(\mc{G})\cap \overline{\tt S}(\mc{G})\Longleftrightarrow$\\ $\qquad\quad
\{\gG < -|\fF|\}\wedge
\{\gG^2+\kK^2 \leq \fF^2 + \pP^2\}
$.
\item $x\in \overline{\tt DNE}(\mc{G})\cap {\tt S}(\mc{G})\Longleftrightarrow$\\ 
$\quad\qquad
\{0 > m \geq -|h|\}\wedge
\{\gG^2+\kK^2 > \fF^2 + \pP^2\}
.$
\item $x\in \overline{\tt DNE}(\mc{G})\cap \overline{\tt S}(\mc{G})\Longleftrightarrow$\\
$\quad\qquad\{
\{0 > \gG \geq -|\fF|\}\wedge
\{\gG^2+\kK^2 \leq \fF^2 + \pP^2\}\}\vee\{\gG \geq 0\}
$.
\end{enumerate}
\end{theorem}

The contributions to the stability of a non-Nash equilibrium or the instability of a Nash equilibrium are 
stated in (ii) and (iii).
We illustrate the geometry of these two cases 
with the shaded regions 
in Fig.~\ref{fig:gkfp:stable&nash}. 
\ifsixpages
First, the differential Nash equilibrium condition is based on the magnitude of $\fF$, implying that
$m \pm zi \in \mb{C}_-^\circ$.
Therefore, if $J$ is not stable at a differential Nash equilibrium, then 
$(\gG , \kK)$ must lie
 inside the circle with radius $\sqrt{\fF^2 + \pP^2}$. 
Second, if a non-Nash equilibrium has stable dynamics, then it must be the case that
 $m\geq |h|$, 
 $m\pm zi \in \mb{C}_-^\circ$, and 
$(m,z)$ is outside the circle. 
\fi

\subsection{Stability: Non-Uniform Learning Rates}
\label{sec:learningrates}

Consider players updating their actions according to gradient play as defined in~\eqref{eq:gradientplay} with
individual learning rates $\learnrate_1,\learnrate_2>0$, not necessarily equal.
We study how the players' {ratio} 
$\tau=\gamma_2/\gamma_1$ 
affects the stability of fixed point $x$
under the learning dynamics by analyzing the game Jacobian 
\begin{equation}
J(x)= 
\bmat{a & b \\ \tau c & \tau d}.
\label{eq:game-jacobian-perturbed}
\end{equation}
Learning rates do not affect whether a fixed point is a Nash equilibrium.
They do, however, affect whether it is stable.

\begin{corollary}[Stability in General-Sum Scalar Games]
Consider a game $\mc{G}=(f_1,f_2)$ on
$X_1\times X_2\subseteq \R^2$ and a fixed point ${x}$. %
Suppose players perform gradient play~\eqref{eq:gradientplay} with
learning rate ratio $\tau=\gamma_2/\gamma_1$.
Then, 
the following are true.
\begin{enumerate}
[label=(\roman*)]

    \item If a Nash equilibrium is stable for some 
    $\tau$, then it is stable for all 
    $\tau$.  
    \item If a non-Nash equilibrium is stable,
    then there exists some 
    $\tau$ that makes it unstable.  
    \item 
    If a 
    fixed point is 
    non-Nash, the determinant of its game Jacobian is positive and $m<|h|$,
    then there exists some 
    $\tau$ that makes it
    stable.  
\end{enumerate}
\label{prop:2x2:learnrate}
\end{corollary}
\begin{proof}
To prove (i), we observe that if $\gG < -|\fF|$, then $\gG \leq \beta \fF$ for all $\beta$ such that $|\beta|<1$.  Choose $-1 \leq \beta = \tfrac{\tau - 1}{\tau + 1} \leq 1$ for $\tau \geq 0$. 
  To prove (ii), 
  choose
    $\tau <  
    \left|\tfrac{a}{d}\right|$.
    Without loss of generality, assume $a < 0$ and $d > 0$. 
    Then, it directly follows that $a + \tau d < 0$.  
To prove (iii), note that
a matrix $J$ is stable if and only if the determinant of $J$ is positive and $m<0$. Hence, without loss of generality, let $d<0$. Then there is a learning rate $\tau$ such that $\tau|d|>|a|$ so that $m<0$.

\end{proof}

Stable Nash equilibria 
in scalar games 
are robust to variations in learning rates and non-Nash equilibria are not.
For continuous games with vector action spaces,
Corollary~\ref{prop:2x2:learnrate}(i) 
no longer holds,
demonstrating that Nash equilibria are not robust, in general,
to variations in learning rates. 
\ifsixpages
Stronger conditions, however, can be stated for the special cases of potential and zero-sum games, corresponding to the symmetric and asymmetric components of the game Jacobian.
\fi

\section{An Illustrative Example}
\label{sec:examples}

We demonstrate our main results
below and in
Fig. \ref{fig:examples}.  
\begin{example}[Nonlinear torus game]
\label{ex:torus}
Consider a game $\mc{G}=(f_1,f_2)$ defined on $\mb S^1\times \mb S^1$ with costs
\eqn{
f_1(x,y) &= \tfrac{2}{a}\cos\left(\tfrac{a}{2} x\right) + \tfrac{2}{a}\cos\left(\tfrac{a}{2}x+by\right),\\
f_2(x,y) &= \tfrac{2}{d}\cos\left(\tfrac{d}{2}y\right) + \tfrac{2}{d}\cos\left(\tfrac{d}{2} y+cx\right).
}
There is a fixed point of the learning dynamics at the origin. Its linearized game Jacobian is 
$
J(0) = \left[\begin{smallmatrix}a & b \\ c & d\end{smallmatrix}\right].
$
First, to show Corollary \ref{prop:2x2:learnrate}(i), we start with an unstable, Nash fixed point of a potential game ($a=-0.4,b=1,c=1,d=-1$). 
We decrease $p = \tfrac{1}{2}(b+c)$ until it becomes stable ($b=0.2,c=0.2$).
Then, we decrease $\tau$ from $1$ to $0.1$ while maintaining stability. 
Second, to show Corollary \ref{prop:2x2:learnrate}(ii), we start with an unstable, non-Nash fixed point of a zero-sum game ($a=0.4,b=-0.2,c=0.2,d=-1$). 
We increase $z=\tfrac{1}{2}(c-b)$ until it becomes stable  ($b=-1,c=1$). 
Then, we decrease $\tau$ from $1$ to $0.01$ making it unstable again. 
Third, to show Corollary \ref{prop:2x2:learnrate}(iii), we start with an unstable, non-Nash fixed point of a Hamiltonian game ($a=0.5,b=0.1,c=0.5,d=-0.5$). 
We increase the interaction term $z=\tfrac{1}{2}(c-b)$ until it becomes marginally stable ($b=-0.5,c=1.1$). 
Then, we increase $\tau$ slightly from $1$ to $2$, making the fixed point stable.

\end{example}

\section{Conclusion}
\label{sec:conclusion}

We characterize 
the local stability 
and Nash optimality of fixed points of 2-player general-sum 
gradient learning dynamics.
\ifsixpages
We assess the contribution of the interaction terms of the game Jacobian to stabilizing/destabilizing a Nash/non-Nash equilibrium.
\fi
Our results give valuable insights 
into the interaction of algorithms in settings
most accurately 
modeled 
as games,
for example, 
when agents lack trust or reliable communication.
In the sequel, we 
characterize continuous games defined on vector action spaces.  
\ifsixpages
We show that all Nash equilibria of zero-sum games are stable and all stable fixed points of potential games are Nash. 
Future work will include more detailed characterizations of the spectrum of general-sum game dynamics and 
analysis of other models of boundedly-rational learning agents.
\fi

\bibliographystyle{plain}
\bibliography{refs}

\begin{figure}[t]
\input{figs/examples.tex}
\caption{
\emph{Demonstration of Corollary~\ref{prop:2x2:learnrate}:}
vector field plots of the three scenarios from Example~\ref{ex:torus}.
}
\label{fig:examples}
\end{figure}

\end{document}